\documentclass[journal,10pt,twocolumn]{IEEEtran}

\usepackage{amsmath,amssymb,amsfonts}
\usepackage{mathtools}
\usepackage{bm}
\usepackage{enumitem}
\usepackage{hyperref}
\usepackage[capitalise,noabbrev]{cleveref}
\usepackage{tikz}
\usetikzlibrary{arrows.meta,positioning,calc}
\usepackage{pgfplots}
\pgfplotsset{compat=1.18}
\usepackage{booktabs}
\usepackage{bbm}
\usepackage{amsthm}
\usepgfplotslibrary{fillbetween}
\usepgfplotslibrary{groupplots}
\usepackage{array}
\usepackage{xcolor}

\newtheorem{theorem}{Theorem}
\newtheorem{lemma}{Lemma}
\newtheorem{corollary}{Corollary}
\newtheorem{proposition}{Proposition}
\newtheorem{remark}{Remark}
\newtheorem{definition}{Definition}
\newtheorem{assumption}{Assumption}

\newcommand{\calA}{\mathcal{A}}
\newcommand{\calK}{\mathcal{K}}
\newcommand{\E}{\mathbb{E}}

\DeclareMathOperator*{\argmax}{arg\,max}
\renewcommand{\Pr}{\mathbb{P}}
\newcommand{\Kmin}{K_{\min}}

\begin{document}

\title{AoI-Aware Polar Code Blind Rate Adaptation}

\author{
  \IEEEauthorblockN{Hassan Tavakoli, Member, IEEE,\\}
  \IEEEauthorblockA{School of EECS,
                    Oregon State University\\
                    Corvallis, OR 97331, USA\\
                    tavakolh@oregonstate.edu}
}

\maketitle

\begin{abstract}
We study adaptive prefix-length selection for reliability-ordered
coded transmission over $N$ coordinates, using only the $K$ most reliable, to balance age of information (AoI) against
throughput. An additive-increase/decrease (AID) prefix-walk controller updates the prefix length
from the last decoding outcome alone, at $O(1)$ per-slot cost with no
channel estimation, inducing an irreducible birth--death chain whose
mean drift fixes a unique balance point $K^\dagger$ (largest $K$ with success probability $\ge\tfrac12$), around which the
stationary distribution concentrates geometrically. We derive this
distribution in closed form via detailed balance and obtain exact
steady-state AoI and goodput via renewal-reward analysis.
Simulations on a $64$-coordinate polar code over a binary symmetric channel (BSC)
confirm these predictions to within $0.3\%$. Relative to a genie-aided
upper bound with perfect channel knowledge, we further quantify the
goodput cost of adapting from a single feedback bit alone.
\end{abstract}


\section{Introduction}
\label{sec:intro}

Perfect channel knowledge requires a validated model or
repeated measuring when the channel changes. Reliability-ordered
polar prefixes make one-bit adaptation analyzable: nested
information sets reduce online rate
adaptation to a one-bit update on $K$, inducing a finite-state
birth--death chain with closed-form stationary distribution, balance
point, steady-state AoI, and goodput.
The update rule itself is structurally related to
classical additive-increase/decrease control (e.g., TCP); what is new
is the closed-form analysis it admits once coupled to nested,
reliability-ordered prefixes, so our results hold at every finite $N$.

Under a monotone success-probability condition, the proposed AID
prefix update induces an irreducible finite-state Markov chain with a
closed-form stationary distribution. This characterization yields
explicit steady-state expressions for goodput and AoI, together with
a geometric concentration result around the operating point where the
success probability is approximately one half.
Let $\calK\triangleq\{\Kmin,\ldots,N\}$ be the allowed prefix lengths, $\Kmin$ the smallest (we write $\{a{:}b\}$ for $\{a,\ldots,b\}$), and $\theta_K$ the success probability at prefix length $K$.
We benchmark AID against the goodput-optimal static prefix
$K^\star \triangleq \argmax_{K\in\calK} (K/N)\theta_K$, a genie-aided
upper bound requiring full offline knowledge of $\{\theta_K\}$; it
quantifies the cost of removing that requirement, not a target AID is
designed to match.
This connects to three established lines of work. AoI, introduced in \cite{kaul2012real} and surveyed in \cite{yates2021age}, has adaptive-scheduling variants that cast update policies as Markov decision processes (MDPs) over the age process with a fixed code \cite{sun2017update,baknina2018coded}. Polar-code rate selection uses the Bhattacharyya ordering for fixed-rate at a target error \cite{arikan2009channel,tal2013how,trifonov2012efficient}, without an online controller.
Recent work has also explored alternative polarization mechanisms and
finite-length polar-code designs with additional structural or privacy
guarantees \cite{MACPolar,RankGuardPolar,GoodIndexRelay}.

\section{System Model}
\label{sec:model}

Time is slotted. A fixed reliability ordering $\pi$ is computed offline; for polar codes, $\pi$ sorts synthetic channels by increasing Bhattacharyya parameter \cite{arikan2009channel,trifonov2012efficient}, and $N=2^n$ for integer $n\ge1$; $n$ denotes the number of polarization stages in the Ar{\i}kan kernel construction~\cite{arikan2009channel}, and $N$ is the resulting maximum blocklength.
At each $t$, the transmitter selects an integer prefix length $K_t\in\calK$ and transmits with information set $\calA_{K_t}\triangleq\{\pi(1),\ldots,\pi(K_t)\}$ at rate $R(K_t)\triangleq K_t/N$, observing the decoding outcome $S_t\in\{0,1\}$ ($1$: success).
Any code family exposes a per-coordinate reliability score that could in principle support prefix-length control. Polar codes combine (i)~a deterministic, offline reliability ordering, (ii)~nested information sets $\calA_{K_1}\subseteq\calA_{K_2}$ for $K_1\le K_2$, and (iii)~a resulting one-dimensional rate control variable $K$. Property~(ii) is what guarantees the update $K_{t+1}=f(K_t,S_t)$ in Definition~\ref{def:rule} always corresponds to a well-defined coordinate set.
Polar codes suit prefix-length control since their synthetic channels admit an offline ordering by Bhattacharyya parameter. Let $Z_{\pi(1)} \le \cdots \le Z_{\pi(N)}$ be the sorted reliabilities; the code rate is then controlled by the single integer $K$, which a one-bit AID rule exploits to adjust rate online without channel estimation.

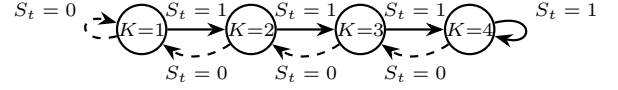
\begin{figure}[t]
\centering
\resizebox{0.9\linewidth}{!}{%
\begin{tikzpicture}[
  >=Stealth,
  node distance=1.4cm,
  state/.style={circle,draw,thick,minimum size=5mm,inner sep=0pt,font=\small},
  success/.style={->,thick},
  failure/.style={->,thick,dashed}
]
  \node[state] (k1) {\scriptsize $K{=}1$};
  \node[state,right of=k1] (k2) {\scriptsize $K{=}2$};
  \node[state,right of=k2] (k3) {\scriptsize $K{=}3$};
  \node[state,right of=k3] (k4) {\scriptsize $K{=}4$};

  \draw[success] (k1) -- node[above,font=\scriptsize] {$S_t=1$} (k2);
  \draw[success] (k2) -- node[above,font=\scriptsize] {$S_t=1$} (k3);
  \draw[success] (k3) -- node[above,font=\scriptsize] {$S_t=1$} (k4);

  \draw[failure] (k2) to[bend left=35] node[below,font=\scriptsize] {$S_t=0$} (k1);
  \draw[failure] (k3) to[bend left=35] node[below,font=\scriptsize] {$S_t=0$} (k2);
  \draw[failure] (k4) to[bend left=35] node[below,font=\scriptsize] {$S_t=0$} (k3);

\draw[failure] (k1) edge[loop left,looseness=8] node[above left,font=\scriptsize] {$S_t=0$}  (k1);
\draw[success] (k4) edge[loop right,looseness=8] node[above right,font=\scriptsize] {$S_t=1$}  (k4);
\end{tikzpicture}}
\caption{A walk over $N{=}4$ ($S_t{=}1/0$: success/failure); rate: $R(K)=K/4$.}
\label{fig:n4_example}
\end{figure}

\paragraph*{Example: Adaptive Polar Prefix Length for $N=4$}
For $N=4=2^2$, the polar transform is generated by the Arikan kernel $F=\begin{bmatrix}1&0\\1&1\end{bmatrix}$ via $G_N=F^{\otimes n}$,
\(
  G_4 = F^{\otimes 2},
\)
where $\otimes$ denotes the Kronecker product.
The codeword is $x=u G_4$. For this example the reliability ordering map $\pi:\{1,2,3,4\}\to\{1,3,2,4\}$ is
\(
Z_1 \le Z_3 \le Z_2 \le Z_4.
\)
The $K$ rows of $G_4$ retained are exactly the rows indexed by the prefix $\calA_K$ of this ordering; the remaining $N-K$ rows correspond to frozen inputs set to zero. Because $\calA_K\subset\calA_{K+1}$ for every $K$ (\Cref{prop:nested}), this prefix is well-defined independent of how $K$ is reached.
The AID update \eqref{eq:AID} for \(N=4\) and \(K_{min}=1\).
Starting from $K_0=1$, the feedback sequence $S=\{1,1,0,1\}$, gives the trajectory $1\,(R{=}1/4)\to2\,(R{=}1/2)\to3\,(R{=}3/4)\to2\,(R{=}1/2)\to3\,(R{=}3/4)$ (\cref{fig:n4_example}).

\subsection{AoI Dynamics and Success Probabilities}
Let $\Delta_t\in\{0,1,2,\ldots\}$ denote the age at the start of slot $t$, where $S_t$ is fed back with negligible delay before the next transmission. The channel itself is memoryless: conditioned on $K_t$, $S_t\sim\mathrm{Bernoulli}(\theta_{K_t})$ 
independently of the past. The unconditional sequence $\{S_t\}$ is nonetheless generally correlated across $t$, since $K_t$ is itself updated from prior outcomes via \eqref{eq:AID}; this dependence resides entirely in the controller's state $K_t$, not the channel.
Equation~\eqref{eq:aoi_dynamics} assumes a fixed-duration slot
regardless of $K$, not modeling transmission-time cost at low rate. The recursion is affine: failure
increments age by one, success resets to zero -- the standard
discrete-time AoI convention~\cite{kaul2012real,yates2021age}. Then
\begin{equation}
  \label{eq:aoi_dynamics}
  \Delta_{t+1} =
  \begin{cases}
    0,            & S_t = 1,\\
    \Delta_t + 1, & S_t = 0.
  \end{cases}
\end{equation}

For each $K\in\calK$, let $\theta_K\triangleq\Pr(S_t=1\mid K_t=K)$, so $S_t\sim\mathrm{Bernoulli}(\theta_K)$ given $K_t=K$. The $\{\theta_K\}$ are fixed by the channel but unknown to the controller.

\begin{assumption}[Monotone success probabilities]
\label{ass:monotone}
\(
  1 > \theta_{\Kmin} \ge \theta_{\Kmin+1} \ge \cdots \ge \theta_N > 0.
\)
\end{assumption}

\begin{remark}
\label{rem:assumption_scope}
Assumption~\ref{ass:monotone} is the \emph{only} substantive
assumption we make. Irreducibility, positive recurrence, reversibility,
and the closed-form stationary distribution
(Theorems~\ref{thm:irreducible}--\ref{thm:stationary}) are
\textbf{derived consequences}, not separately assumed. The physical
channel need only produce binary success/failure outcomes. The success
profile $\{\theta_K\}$ is fixed but unknown to the transmitter; AID
locates an operating point within it using one feedback bit per slot,
without estimation. Extending to time-varying $\theta_K(t)$ is future
work (\Cref{sec:conclusion}).
\end{remark}

This holds under the finite-blocklength (FBL) approximation \cite{polyanskiy2010channel}: for rate $R(K)=K/N$,
\begin{equation}
  \label{eq:fbl}
  \epsilon=\epsilon(N,R) \approx Q\!\left(\frac{C - R + \frac{\log_2 N}{2N}}{\sqrt{V/N}}\right),
\end{equation}
with $C,V$ the capacity and dispersion and $Q$ the Gaussian tail function. Defining the FBL surrogate $\tilde\theta_K\triangleq 1-\epsilon(N,K/N)$.

\begin{lemma}
\label{lem:fbl}
$\tilde\theta_K$ is strictly decreasing in $K$, so Assumption~\ref{ass:monotone}
holds with $\tilde\theta_K$ in place of $\theta_K$. The goodput
$\tilde G(K)\triangleq(K/N)\tilde\theta_K$ is positive on $\calK$ and attains a
maximizer $K^\star\in\calK$.
\end{lemma}
\begin{proof}
The argument of $Q(\cdot)$ in \eqref{eq:fbl} is strictly decreasing in
$K$ and $Q$ is strictly decreasing, so $\epsilon$ is strictly increasing and
$\tilde\theta_K=1-\epsilon$ is strictly decreasing. Since $Q(x)\in(0,1)$,
$\tilde G(K)>0$ on the finite set $\calK$, so a maximizer exists.
\end{proof}

\subsection{AID Prefix-Walk Controller and Fixed Baseline}
\label{sec:AID}

\begin{definition}
\label{def:rule}
Fix $\Kmin\ge1$ and $N$. At each slot $t$: transmit with $\calA_{K_t}$, observe $S_t$, and update; AoI evolves by \eqref{eq:aoi_dynamics}.
\begin{equation}
  \label{eq:AID}
  K_{t+1} =
  \begin{cases}
    \min(K_t + 1,\, N),     & S_t = 1,\\
    \max(K_t - 1,\, \Kmin), & S_t = 0.
  \end{cases}
\end{equation}
\end{definition}
The $\pm1$ step is the parameter-free choice: for steps $+a/-b$ the zero-drift point is $\theta_K=b/(a+b)$, so $a=b$ targets $\tfrac12$ without channel knowledge, while $a\neq b$ implicitly selects a target; $a=b=1$ is also the only nearest-neighbor (birth--death) case, giving \eqref{eq:stationary}.
An AoI-dependent step would make $(\Delta_t,K_t)$ the relevant state.
The rule requires one Bernoulli observation and one integer increment/decrement per slot, $O(1)$ cost, and no estimate of $\theta_K$ is consulted. The rule is \emph{channel-blind} by construction. Since $K_{t+1}$ is a deterministic function of $K_t$ and the single random bit $S_t$, the process $\{K_t\}$ is a time-homogeneous Markov chain on $\calK$.

\begin{proposition}
\label{prop:nested}
If $K_{t_1}\le K_{t_2}$ then $\calA_{K_{t_1}}\subseteq\calA_{K_{t_2}}$.
\end{proposition}
\begin{proof}
$\pi$ is fixed, and $\calA_K=\{\pi(1):\pi(K)\}$ is a prefix of $\pi$ for every $K$; prefixes of the same ordering are nested.
\end{proof}

\begin{remark}[Memoryless channel, stateful controller]
\label{rem:memoryless_stateful}
The channel is memoryless given $K_t$; correlation in $\{S_t\}$ comes
only from the controller's own state $K_t$. AID exploits no channel
memory: its value is locating $K^\dagger$ without prior knowledge of
$\{\theta_K\}$, unlike a fixed policy, which needs $\{\theta_K\}$ to
set $K^\star$. Maximizing $\theta_K$ alone favors $K=\Kmin$, but AID
tracks the drift sign $2\theta_K-1$ (Proposition~\ref{prop:drift}),
not $\theta_K$ itself, so $K^\dagger\neq\Kmin$; always using $\Kmin$
would minimize AoI at the cost of near-zero throughput, the corner
case this joint operating point avoids. Table~\ref{tab:mismatch}
confirms this: under mismatch, AID's goodput holds up while the
fixed baseline's degrades, even surpassing it at BSC crossover probability $p=0.16$.
\end{remark}

As a non-adaptive reference, we assume the static rule $K_t\equiv K^\star$ for all $t$, where $K^\star$ is the goodput-maximizing prefix length under perfect knowledge of $\{\theta_K\}$. 
Under this static rule only, $S_t$ is i.i.d.\ Bernoulli$(\theta_{K^\star})$ since the transmitted rate never changes; the steady-state goodput is exactly
\(
  G^\star = \frac{K^\star}{N}\theta_{K^\star},
\)
the steady-state AoI follows from a geometric inter-success time with parameter $\theta_{K^\star}$ as
\(
  \Delta^\star = \frac{1-\theta_{K^\star}}{2\,\theta_{K^\star}}.
\)
This offline oracle isolates the cost of trading exact channel knowledge for a one-bit, estimation-free feedback rule.

\section{Closed-Form Performance with Birth--Death}
\label{sec:descent}

\begin{proposition}[Mean drift]
\label{prop:drift}
Under Assumption~\ref{ass:monotone}, for $K\in\{\Kmin+1:N-1\}$,
\(
  \E[K_{t+1}-K_t\mid K_t=K] = 2\theta_K-1.
\)
\end{proposition}
\begin{proof}
By \eqref{eq:AID}, $K_{t+1}-K_t=+1$ with probability $\theta_K$ and $-1$ with probability $1-\theta_K$; taking the expectation gives $(+1)\theta_K+(-1)(1-\theta_K)=2\theta_K-1$.
\end{proof}

We also define the balance point as 
$K^\dagger \triangleq \max\{K\in\calK : \theta_K \ge \tfrac{1}{2}\}$.
By Proposition~\ref{prop:drift}, drift is positive for $K<K^\dagger$, negative for
$K>K^\dagger$, making $K^\dagger$ the unique fixed point -- and
explaining why AID need not locate $K^\star$: the drift depends only on
whether $\theta_K\ge\tfrac12$, not on the goodput value $(K/N)\theta_K$.

\begin{lemma}
\label{lem:plays_infinite}
Under Definition~\ref{def:rule} and Assumption~\ref{ass:monotone}, every $K\in\calK$ is visited infinitely often almost surely.
\end{lemma}
\begin{proof}
$\{K_t\}$ is a random walk on the graph $\{\Kmin:N\}$ with strictly positive up/down transition probabilities at every interior state ($\theta_K\in(0,1)$) and inward drift at the boundaries. A finite irreducible chain is positive recurrent \cite{norris1998markov}.
\end{proof}

\begin{theorem}
\label{thm:irreducible}
Under Assumption~\ref{ass:monotone}, $\{K_t\}$ is irreducible and positive recurrent on $\calK$, with a stationary distribution $\mu$.
\end{theorem}
\begin{proof}
Strictly positive transition probabilities throughout $\calK$ give irreducibility on a finite set, which implies positive recurrence and uniqueness of $\mu$ \cite{norris1998markov}.
\end{proof}

\begin{theorem}
\label{thm:stationary}
The unique stationary distribution of $\{K_t\}$ is
\begin{align}
  \label{eq:stationary}
  \mu(K) = \mu(\Kmin)\prod_{j=\Kmin}^{K-1}\frac{\theta_j}{1-\theta_{j+1}},
\end{align}
and $\Pr(K_t=K)\to\mu(K)$ as $t\to\infty$ from any initial state.
\end{theorem}
\begin{proof}
A finite birth--death chain is reversible \cite{norris1998markov}, so $\mu$ satisfies $\mu(K)\theta_K=\mu(K{+}1)(1-\theta_{K+1})$; solving recursively from $K=\Kmin$ gives \eqref{eq:stationary}, with $\mu(\Kmin)= \bigl(1+\sum_{K=\Kmin+1}^{N}\prod_{j=\Kmin}^{K-1}\frac{\theta_j}{1-\theta_{j+1}}\bigr)^{-1}$ fixed by normalization. Convergence follows from irreducibility and aperiodicity (the latter from boundary self-loops).
By \eqref{eq:stationary}, $\mu(K{+}1)/\mu(K)=\theta_K/(1-\theta_{K+1})$
is at least one for $K<K^\dagger$ and below one for $K>K^\dagger$, so $\mu$ is unimodal with mode at $K^\dagger$ or $K^\dagger{+}1$ (at $K^\dagger$ iff $\theta_{K^\dagger}+\theta_{K^\dagger+1}\le1$).
\end{proof}

\begin{corollary}[Geometric concentration around $K^\dagger$]
\label{cor:concentration}
For every integer $m\ge1$ and \(\delta_R\triangleq\log\frac{1-\theta_{K^\dagger+1}}{\theta_{K^\dagger+1}}>0\) (strictly positive by Assumption~\ref{ass:monotone}, but $\delta_R\to0$ as $\theta_{K^\dagger+1}\to\tfrac12^-$),
\begin{equation}
  \label{eq:concentration}
  \Pr_\mu(K>K^\dagger+m)\le C_\gamma\exp(-m\delta_R),
\end{equation}
where $r\triangleq\theta_{K^\dagger+1}/(1-\theta_{K^\dagger+1})$, $\gamma\triangleq\tfrac12-\theta_{K^\dagger+1}$, 
and $C_\gamma\triangleq 1/(1-r)$ with a left-tail bound
\(
  \Pr_\mu(K < K^\dagger - m)
  \;\le\;
  \frac{e^{-(m+1)\delta_L}}{1-e^{-\delta_L}},
  \,
  \delta_L \triangleq \log\frac{\theta_{K^\dagger}}{1-\theta_{K^\dagger}} \ge 0,
\)
whenever $\theta_{K^\dagger}>\tfrac12$.
\end{corollary}
\begin{proof}
By Assumption~\ref{ass:monotone}, $\theta_{K^\dagger+1}<\tfrac12$ (since $K^\dagger$ is the last index with $\theta_K\ge\tfrac12$), so $r<1$ strictly. 
By \eqref{eq:stationary}, $\mu(K{+}1)/\mu(K)\le r$ for $K\ge K^\dagger{+}1$, so $\mu(K^\dagger{+}1{+}i)\le r^i\mu(K^\dagger{+}1)$ for $i\ge0$ and
\(
  \Pr_\mu(K>K^\dagger+m)\le\sum_{i\ge m}r^i=\frac{r^m}{1-r}=C_\gamma e^{-m\delta_R}.
\)
For $K<K^\dagger$, $\mu(K)/\mu(K{+}1)=(1-\theta_{K+1})/\theta_K\le e^{-\delta_L}$, so $\mu(K^\dagger-j)\le e^{-j\delta_L}\mu(K^\dagger)$; summing over $j\ge m+1$ gives the left tail.
The bound is valid for every $K^\dagger$ under Assumption~\ref{ass:monotone}, but its informativeness depends on the margin $\gamma\ge0$: as $\gamma\to0$, $r\to1$ and $C_\gamma\to\infty$, so the bound, becomes less informative -- a reflection of how close the channel's profile places $\theta_{K^\dagger+1}$ to one-half, not a defect of the bound.
\end{proof}

Corollary~\ref{cor:concentration} is a finite-state property of
$\mu$: the limit involved is only $t\to\infty$ (Theorem~\ref{thm:irreducible}),
not $N\to\infty$, so \eqref{eq:concentration} needs no large-$N$
requirement.
Since $\Delta_t$ depends on the full history since the last reset, the
AoI analysis uses the regenerative process $(\Delta_t,K_t)$, reset at
every success.

\subsection{Regenerative Structure and Hitting Times}

Let $\tau_0<\tau_1<\cdots$ be the success epochs, $L_n\triangleq\tau_{n+1}-\tau_n$ the inter-success time, and $R_n\triangleq L_n(L_n-1)/2$ the total age accrued in cycle $n$ (ages $0,1,\ldots,L_n-1$). With $k^-\triangleq\max(k-1,\Kmin)$, define $h_1(k)\triangleq\E[L_n\mid K_{\tau_n+1}=k]$ and $h_2(k)\triangleq\E[L_n^2\mid K_{\tau_n+1}=k]$.

\begin{lemma}[Hitting-time recursions]
\label{lem:hitting}
\begin{align}
  \label{eq:h1}
  h_1(k) &= 1 + (1-\theta_k)\,h_1(k^-),\\
  \label{eq:h2}
  h_2(k) &= 1 + 2(1-\theta_k)\,h_1(k^-) + (1-\theta_k)\,h_2(k^-),
\end{align}
with $h_1(\Kmin)=1/\theta_{\Kmin}$, $h_2(\Kmin)=(2-\theta_{\Kmin})/\theta_{\Kmin}^2$.
\end{lemma}
\begin{proof}
Condition on the first-slot outcome at state $k$: an immediate success occurs with probability $\theta_k$ ($L=1$); a failure (probability $1-\theta_k$) moves the chain to $k^-$ and adds one slot to the remaining inter-success time, $L=1+M$ with $M$ distributed as $L$ from $k^-$. Taking expectations of $L$ and $L^2$ gives \eqref{eq:h1}--\eqref{eq:h2}. At $\Kmin$, failures self-loop, so $L\sim\mathrm{Geom}(\theta_{\Kmin})$, giving the boundary values.
The cycle-start distribution $\rho$ (the law of $K_{\tau_n+1}$) is the success-weighted distribution
\(\bar\theta\triangleq\sum_{k\in\calK}\mu(k)\theta_k\):
\(
  \rho(k') = \frac{1}{\bar\theta}\sum_{k:\min(k+1,N)=k'}\mu(k)\theta_k,
\)
with $\rho(\Kmin)=0$ and boundary term $\rho(N)=[\mu(N)\theta_N+\mu(N-1)\theta_{N-1}]/\bar\theta$.
\end{proof}

\subsection{Steady-State AoI and Goodput}
\begin{theorem}
\label{thm:aoi_performance}
With $\mu$ from Theorem~\ref{thm:stationary}, and $\rho$, and $h_1,h_2$ from Lemma~\ref{lem:hitting}, define
\(
  \E[L]\triangleq\sum_k\rho(k)h_1(k)
\),
\(
  \E[L^2]\triangleq\sum_k\rho(k)h_2(k).
\)
Then, almost surely from any initial state,
\begin{align}
  \label{eq:ss_goodput}
  \bar{G} &\triangleq \lim_{T\to\infty}\frac{1}{T}\sum_{t=1}^T \frac{K_t}{N}S_t = \sum_{K\in\calK}\mu(K)\frac{K}{N}\theta_K,\\
  \label{eq:ss_aoi}
  \bar{\Delta} &\triangleq \lim_{T\to\infty}\frac{1}{T}\sum_{t=1}^T \Delta_t = \frac{\E[L^2]-\E[L]}{2\E[L]}.
\end{align}
\end{theorem}
\begin{proof}
\textit{Goodput:} $(K_t/N)S_t$ has mean $(K/N)\theta_K$ given $K_t=K$, a bounded function of $K_t$ alone, so the ergodic theorem for finite irreducible chains \cite{norris1998markov} gives \eqref{eq:ss_goodput}.
\textit{AoI:} $(\Delta_t,K_t)$ regenerates at each success epoch since $\Delta_{\tau_n+1}=0$ and the future depends only on $K_{\tau_n+1}$. By the renewal-reward theorem \cite{norris1998markov}, $\bar\Delta=\E[R_n]/\E[L_n]=\E[L_n(L_n-1)/2]/\E[L_n]=(\E[L^2]-\E[L])/(2\E[L])$, with moments $L_n$ computed via \eqref{eq:h1}--\eqref{eq:h2} and averaged over the cycle-start law $\rho$.
\end{proof}

\begin{remark}
\label{rem:upper_bound}
$K^\star$ requires exact knowledge of $\{\theta_K\}$, recomputed
whenever the channel changes; AID uses neither. $G^\star$ is thus a
genie-aided upper bound on any policy, and $G^\star-\bar G$
prices the removal of channel knowledge, quantified under mismatch in
Table~\ref{tab:mismatch}. The present simulations validate this trade-off for a fixed, unknown channel; extending it to bootstrapping or time-varying channels is future work (\Cref{sec:conclusion}).
\end{remark}

\begin{figure}[!t]
\centering
\begin{tikzpicture}
\begin{groupplot}[
  group style={group size=2 by 1, horizontal sep=10mm},
  width=0.5\columnwidth, height=38mm,
  grid=major, grid style={gray!15},
  axis line style={gray!60},
  tick label style={font=\scriptsize},
  label style={font=\scriptsize},
  title style={font=\scriptsize},
]
\nextgroupplot[
  title={Stationary distribution},
  xlabel={Prefix length $K$}, ylabel={$\mu(K)$},
  xmin=-5, xmax=65, ymin=0, ymax=0.21,
  xtick={20,27,35,44,60},
  ybar, bar width=2pt,
  legend style={font=\tiny,at={(0.52,0.95)},anchor=north east,draw=gray!50,fill=white},
  legend cell align=left,
]
\addplot[fill=teal!70!black,draw=none,fill opacity=0.85] coordinates {
  (4,0.000)(8,0.000)(12,0.000)(16,0.001)(20,0.003)(24,0.009)
  (27,0.026)(28,0.032)(29,0.039)(30,0.048)(31,0.058)(32,0.068)
  (33,0.077)(34,0.084)(35,0.179)(36,0.160)(37,0.101)(38,0.060)
  (39,0.034)(40,0.018)(41,0.009)(42,0.005)(43,0.002)(44,0.001)};
\addlegendentry{Theory}
\addplot[fill=orange!85!red,draw=none,fill opacity=0.55] coordinates {
  (4,0.000)(8,0.000)(12,0.000)(16,0.001)(20,0.003)(24,0.009)
  (27,0.026)(28,0.032)(29,0.039)(30,0.048)(31,0.058)(32,0.068)
  (33,0.077)(34,0.084)(35,0.180)(36,0.161)(37,0.101)(38,0.061)
  (39,0.034)(40,0.019)(41,0.009)(42,0.005)(43,0.002)(44,0.001)};
\addlegendentry{Empirical}
\draw[gray!70,dashed] (axis cs:27,0) -- (axis cs:27,0.21) node[above,font=\tiny,black] {$K^\star$};
\draw[gray!70,dashed] (axis cs:35,0) -- (axis cs:35,0.21) node[above,font=\tiny,black] {$K^\dagger$};
\nextgroupplot[
  title={Operating point},
  ylabel={Value},
  symbolic x coords={AoI,Goodput},
  xtick=data,
  enlarge x limits=0.3,
  xticklabel style={font=\scriptsize},
  ymin=0, ymax=1.05,
  ybar, bar width=10pt,
  legend style={font=\tiny,at={(0.97,0.95)},anchor=north east,draw=gray!50,fill=white},
  legend cell align=left,
  nodes near coords,
  nodes near coords style={font=\tiny},
  every node near coord/.append style={anchor=south},
]
\addplot[fill=teal!70!black,draw=none,fill opacity=0.85] coordinates {(AoI,0.9174) (Goodput,0.2696)};
\addlegendentry{AID}
\addplot[fill=gray!55,draw=none,fill opacity=0.85] coordinates {(AoI,0.1689) (Goodput,0.3609)};
\addlegendentry{Fixed $K^\star$}
\end{groupplot}
\end{tikzpicture}
\caption{Left: AID stationary distribution, theory \eqref{eq:stationary} vs.\ empirical occupancy ($30$ trials $\times$ $10{,}000$ slots), mode at $K^\dagger=35$; dashed lines mark $K^\star=27$ and $K^\dagger=35$. Right: steady-state operating point, AID vs.\ genie-aided $K^\star=27$, $N=64$, BSC($p=0.11$). AID's AoI is theory-predicted via Theorem~\ref{thm:aoi_performance}.}
\label{fig:combined}
\end{figure}
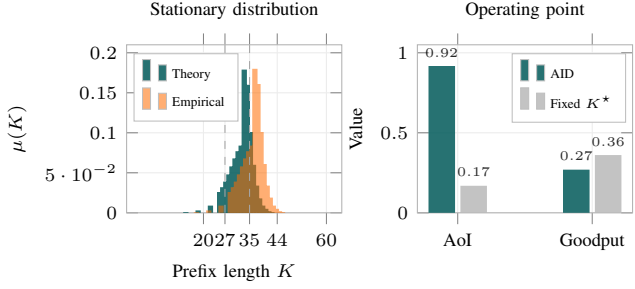

\begin{table}[h]
\centering
\caption{Key operating-point comparison across channel types (AWGN: additive white Gaussian noise; SNR: signal-to-noise ratio): $N=64$, $C=0.5$, $30$ trials $\times$ $10{,}000$ slots.}
\label{tab:main_results}
\setlength{\tabcolsep}{3pt}
\begin{tabular}{lcccc}
\toprule
 & \multicolumn{2}{c}{BSC ($p=0.11$)} & \multicolumn{2}{c}{AWGN (SNR $=1.00$)} \\
Quantity & Theory & Empirical & Theory & Empirical \\
\midrule
AID AoI $\bar\Delta$ (slots) & $0.9174$ & $0.9194$ & $0.9129$ & $0.9149$ \\
AID goodput $\bar G$         & $0.2696$ & $0.2696$ & $0.2696$ & $0.2696$ \\
Genie goodput $G^\star$ & $0.3609$ & --- & $0.3676$ & --- \\
Goodput gap $G^\star-\bar G$ & $0.0913$ & --- & $0.0980$ & --- \\
\bottomrule
\end{tabular}
\end{table}

\begin{table}[h]
\centering
\caption{Goodput under channel mismatch: $K^\star$ designed for
$p=0.11$. AID adapts; fixed $K^\star=27$.
$N=64$, $30$ trials $\times$ $10{,}000$ slots.}
\label{tab:mismatch}
\setlength{\tabcolsep}{4pt}
\begin{tabular}{ccccc}
\toprule
True $p$ & AID $\bar{G}$ & Fixed $K^\star{=}27$ goodput
         & Gap (fixed$-$AID) & AID wins? \\
\midrule
$0.08$ & $0.3184$ & $0.3434$ & $+0.0250$ & No  \\
$0.11$ & $0.2696$ & $0.3609$ & $+0.0913$ & No  \\
$0.16$ & $0.2024$ & $0.1967$ & $-0.0057$ & \textbf{Yes}  \\
\bottomrule
\end{tabular}
\end{table}

\begin{table}[h]
\centering
\caption{AID sensitivity to BSC crossover probability $p$, $N=64$,
         $30$ trials $\times$ $10{,}000$ slots.}
\label{tab:p_sweep}
\setlength{\tabcolsep}{4pt}
\begin{tabular}{ccccc}
\toprule
$p$ & $\bar\Delta$ (theory) & $\bar\Delta$ (sim.)
    & $\bar G$ (theory) & $\bar G$ (sim.) \\
\midrule
0.02 & 0.9034 & 0.9039 & 0.4490 & 0.4490 \\
0.05 & 0.9161 & 0.9179 & 0.3763 & 0.3763 \\
0.08 & 0.9182 & 0.9198 & 0.3184 & 0.3185 \\
0.11 & 0.9174 & 0.9194 & 0.2696 & 0.2696 \\
0.15 & 0.9137 & 0.9137 & 0.2146 & 0.2147 \\
\bottomrule
\end{tabular}
\end{table}

\begin{table}[h]
\centering
\caption{Operating points vs.\ blocklength, BSC($p{=}0.11$), $\Kmin=4$.}
\label{tab:blocklength}
\setlength{\tabcolsep}{3pt}
\begin{tabular}{ccccccc}
\toprule
$N$ & $K^\star$ & $K^\dagger$ & $K^\star/N$ & $K^\dagger/N$ & $G^\star$ & $\bar G$ \\
\midrule
64    & 27    & 35    & 0.4219 & 0.5469 & 0.3609 & 0.2696 \\
256   & 109   & 132   & 0.4258 & 0.5156 & 0.3987 & 0.2569 \\
1024  & 459   & 517   & 0.4482 & 0.5049 & 0.4360 & 0.2522 \\
4096  & 1918  & 2054  & 0.4683 & 0.5015 & 0.4626 & 0.2507 \\
\bottomrule
\end{tabular}
\end{table}

AID's mean rate $\bar R=K^\dagger/N=0.547$ exceeds capacity $C=0.5$
by the dispersion correction $\log_2(N)/(2N)$, vanishing as $N\to\infty$.
\Cref{tab:p_sweep} confirms the agreement of \Cref{tab:main_results}
holds across five BSC qualities, with
$\bar\Delta$ stable near $0.917$ as $K^\dagger$ tracks $\bar\theta\approx0.5$.

\begin{corollary}
\label{cor:asymptotic-rate}
As $N\to\infty$ with $(C,V)$ fixed, $K^\dagger/N,K^\star/N\to C$: the
rate gap vanishes, but $G^\star-\bar G$ need not.
\Cref{tab:blocklength} confirms this: rate gap $0.125\to0.033$ over
$N{=}64$--$4096$; goodput gap stays in $[0.09,0.21]$, limit open.
\end{corollary}
\begin{proof}
By \eqref{eq:fbl}, $\tilde\theta_K\to \mathbf{1}\{K/N<C\}$ as $N\to\infty$, forcing $K^\dagger/N,K^\star/N\to C$; the rates of convergence differ, so $\theta_{K^\star}$ need not track $\theta_{K^\dagger}=\tfrac12$, leaving $G^\star-\bar G$ unconstrained.
\end{proof}

\section{Simulation Results}
\label{sec:simulation}

We simulate AID for $N=64$ under a BSC with $p=0.11$; true success probabilities equal the FBL surrogate $\tilde\theta_K$ of \eqref{eq:fbl}, never revealed to the controller, which observes only $S_t$. Parameters: $\Kmin=4$, $T=10{,}000$ slots per trial, $30$ independent trials, and burn-in of $2{,}000$ slots. The system parameters are $K^\star=27$ and $K^\dagger=35$ (where $\tilde\theta_K=0.5$).
\Cref{fig:combined} presents the two central empirical results in a single figure. The left panel validates Theorem~\ref{thm:stationary}: the closed-form stationary distribution $\mu(K)$ matches the empirical occupancy over $30$ trials to within a maximum pointwise deviation of $0.0011$, with the mode at $K^\dagger=35$. The right panel summarizes the AID's steady-state AoI and goodput against the fixed-prefix baseline at $K^\star$.
\Cref{tab:main_results} summarizes the operating points for BSC and AWGN for SNR $=1.00$, SNR is chosen so AWGN capacity matches the BSC reference $C=0.5$: theory and
simulation agree within $0.22\%$ for AoI and $0.00\%$ for goodput.
\emph{Robustness to channel mismatch}
$K^\star$ is tuned for $p=0.11$; under mismatch (Table~\ref{tab:mismatch}),
its goodput degrades while AID self-corrects via feedback alone. At
$p=0.16$, the fixed baseline's mismatch penalty exceeds AID's structural
gap, making AID the preferred policy under channel uncertainty.

\section{Conclusion}
\label{sec:conclusion}
We give a finite-state characterization of AID: irreducibility and a closed-form stationary distribution, geometric
concentration around $K^\dagger$,
and exact steady-state AoI and goodput via renewal-reward, confirmed to within $0.3\%$ in simulation.
Against the genie-aided bound $K^\star$, AID incurs a
goodput gap by equilibrating at $K^\dagger$ rather than the goodput
optimum; this gap can invert under channel mismatch, favoring AID when the assumed channel is
wrong. It is the structural cost of estimation-free, one-bit
adaptation, available in closed form per $N$, though its
$N\to\infty$ limit remains open.
Limitations: We assume $\{\theta_K\}_{K\in\calK}$
fixed and the channel memoryless, instantaneous
feedback, and a fixed-duration slot regardless of $K$, so
\eqref{eq:aoi_dynamics} omits the transmission-time cost of a
low-rate update -- scope choices enabling a full closed-form
analysis.

\bibliographystyle{IEEEtran}

\end{document}